\documentclass[12pt, english]{article}

\usepackage{babel}
\usepackage[latin1]{inputenc}
\usepackage{amsmath}
\usepackage{amssymb}
\usepackage{amsthm}

\theoremstyle{plain}
\newtheorem{Thm}{Theorem}[section]
\newtheorem{Prop}[Thm]{Proposition}

\newtheorem{Lemma}[Thm]{Lemma}

\newtheorem{faulty}[Thm]{Erroneous Conclusion}

\theoremstyle{definition}

\newtheorem{Def*}{Definition}

\newtheorem{Cor}[Thm]{Corollary}

\begin{document}

\title{When only the last one will do}

\author{Johan W\"astlund \\
\small Department of Mathematical Sciences\\[-0.8ex]
\small Chalmers University of Technology, \\[-0.8ex] 
\small S-412 96 Gothenburg, Sweden\\[-0.8ex]
\small \texttt{wastlund@chalmers.se}
}
\date{\small \today} 

\maketitle

\begin{abstract}  
An unknown positive number of items arrive at independent uniformly distributed times in the interval $[0,1]$ to a selector, whose task is to pick online the last one. We show that under the assumption of an adversary determining the number of items, there exists a game-theoretical equilibrium, in other words the selector and the adversary both possess optimal strategies. The probability of success of the selector with the optimal strategy is estimated numerically to $0.352917000207196$.  
\end{abstract}

%2010 Mathematics Subject Classification: 60G40, 62L15, 91A60.

\section{Introduction}

An unknown number $n\geq 1$ of items are presented at independent uniformly distributed times in the interval $[0,1]$ to a selector, whose task is to choose online the last one. This is related to so-called parking and house-selling problems in the theory of optimal stopping, but has a more adversarial flavor. The problem was originally motivated by a generalization of the secretary problem to partially ordered sets (which we discuss in Section~\ref{S:poset}), but in view of the simplicity seemed worthy of study in its own right. 

We begin with some observations showing that the probability of success with a good strategy is a nontrivial constant. In Section~\ref{S:existence} we show that both the selector and the adversary possess optimal strategies. In Section~\ref{S:computing} we show how to compute numerically the probability of success for the selector at game-theoretical equilibrium. Finally in Section~\ref{S:poset} we explain the origin of the problem and its relation to the partially ordered secretary problem.

\medskip

{\bf Acknowledgments.} This paper has grown out of discussions in coffee breaks at conferences in the last ten years. I wish to thank in particular David Aldous, Thomas Bruss, Ragnar Freij, Svante Janson, Joel Spencer and Peter Winkler.  

\section{Basic observations} \label{S:basic}

Starting from the selector's point of view, we wish to maximize the probability of picking the last item under a worst-case scenario. We first give a simple proof that there is a strategy for the selector that achieves a probability of success bounded away from zero. 

\begin{Prop} \label{P:soft}
There is a strategy by which the selector succeeds with probability at least some constant $c>0$ for every $n$.  
\end{Prop}
\begin{proof} The selector starts by observing the number $k$ of items arriving in $[0,1/2]$. The idea is to estimate the total number of items to roughly $2k$, in which case the remaining arrival times can be approximated by a Poisson point process of rate $2k$. In particular we expect that with probability roughly $1/e$, the interval $[1-1/(2k), 1]$ will contain exactly one item. Our strategy is therefore to wait until time $1-1/(2k)$ and accept the next one. 

Since we do not have an a priori probability distribution on $n$, we cannot talk about the probability distribution of $n$ conditioning on $k$. Instead we must fix the selector's strategy and analyze it for arbitrary $n$.
With the trivial modification of accepting the first item if $k=0$, the strategy will clearly achieve positive probability of success for each $n$, and therefore it suffices to analyze it for large $n$. 

For $\epsilon>0$ we have, by the law of large numbers, \begin{equation} \label{nk} \left(\frac12 - \epsilon\right)n < k <\left(\frac12 + \epsilon\right)n\end{equation} with high probability as $n\to\infty$. Conditioning on $k$ and supposing that \eqref{nk} holds, the probability that exactly one of the remaining $n-k$ items arrives after time $1-1/(2k)$ is $$(n-k)\cdot\frac1k\cdot \left(1-\frac1k\right)^{n-k-1} = (1+O(\epsilon))\cdot e^{-1}.$$

Therefore as $n\to\infty$, the probability of success will converge to $1/e$, which establishes the claim.
 \end{proof}
 
Thomas Bruss has pointed out that the \emph{odds-algorithm} \cite{B00, B03} provides a quite good heuristic for the selector. Although the criteria of the main theorem of \cite{B00} are not met, it suggests a strategy which accepts the $k$:th item if it arrives after time $1-1/(k+1)$. This strategy gives a success probability of at least $5/16$, the minimum occurring for $n=3$, but we have no computation-free proof of this.

The proof we have given here of Proposition~\ref{P:soft} essentially reduces the problem to choosing online the last event of a Poisson process of known rate. The idea of treating the arrival times as a Poisson point process also forms the basis for an informal argument that shows (in a way which can be made rigorous) that the selector cannot achieve uniformly in $n$ a better success probability than $1/e$. 

We think of the number $n$ as chosen by an adversary, the devil. The devil can choose $n$ from any probability distribution, and in particular may choose $n$ according to Poisson($\lambda$)-distribution. This means that the arrival times will be the times of the events in a rate $\lambda$ Poisson point process on $[0,1]$. There is a small technicality to resolve: The devil is strictly not allowed to choose $n=0$, but we can modify the rules of the game so that the selector wins by default if $n=0$.

The selector wins by default with probability $e^{-\lambda}$, and must otherwise choose online the last event of the Poisson process whose rate $\lambda$ we can even assume to be known to the selector. We will present a more general argument in detail later, but for the moment let us assume that the only reasonable thing the selector can do is to choose a point in time and decide to accept the next item. Then the selector wins if the number of remaining items is exactly 1. If the expected number of remaining items is $x$, then the probability of exactly one is $xe^{-x}$, which attains a maximum of $1/e$ for $x=1$. If our assumption about the strategy of the selector is correct, the devil can therefore keep the selector's winning probability below any number which is greater than $1/e$.

At this point we present an erroneous argument which pretends to demonstrate that, as for the classical secretary problem as well as the setting of \cite{B00}, a success probability of $1/e$ is achievable. This caused me some confusion, especially at the point when another argument seemed to establish the opposite.
The way it is presented here the argument lacks rigor in several respects, but as we will later see, the frivolous assumption of the existence of a game-theoretical equilibrium is not the most serious issue. Although incorrect, the argument may be of some interest as it reveals a certain paradoxical nature of the problem. I apologize to readers who are unable to follow it, but explaining it in greater detail is actually impossible.

\begin{faulty} \label{T:faulty}
There is a strategy for the selector that succeeds with probability at least $1/e$ regardless of $n$.
\end{faulty}

\begin{proof}[``Proof'']
First observe that an optimal strategy for the devil cannot be to choose $n$ from a distribution with finite support. If it were, then the selector's optimal strategy would \emph{always} accept the $N$:th item, where $N$ is the largest value of $n$ chosen with positive probability. But then the devil could switch strategy and instead choose $n=N+1$ and always win. 

Hence the devil will choose arbitrarily large values of $n$ with positive probability. But if $n$ is large, the strategy described in the proof of Proposition~\ref{P:soft} will achieve winning probability tending to $1/e$. Since at game-theoretical equilibrium, all values of $n$ in the devil's strategy must be equally likely to win, the selector's winning probability must be at least $1/e$.
\end{proof}

After this introduction, where our arguments have deteriorated from sub-optimal to non-rigorous to incorrect, let us start over again and be more careful.

\section{Existence of optimal strategies} \label{S:existence}

\subsection{A card game}
We wish to apply the von Neumann-Morgenstern equilibrium theorem, and therefore start by studying a similar game where each player has only finitely many pure strategies. 

In this game there is a deck of $d$ cards, and at the start of the game the devil labels the face side of some of them (labels representing items). There must be at least one labeled card, and at most $N$. The deck is then shuffled and the cards are turned up one by one. The selector, knowing $d$ and $N$, must select online the last labeled card.

The card game has only finitely many pure strategies for each of the players, and therefore by the von Neumann-Morgenstern theorem has a strategic equilibrium.
Next we establish some properties of the equilibrium strategies. We assume throughout that $d\geq N$. Notice that the devil's strategy is simply a probability distribution on $\{1,\dots,N\}$. 

\begin{Lemma} At any equilibrium in the card game, the devil's strategy gives positive probability to each of the numbers $1,\dots, N$. 
\end{Lemma}

\begin{proof}

First notice that at equilibrium, the selector must win with positive probability, since there is a trivial strategy that wins with probability $1/N$: Just make a guess on the number of labeled cards from uniform distribution on $\{1,\dots, N\}$.

It follows that at equilibrium, the selector's strategy must have the property that for every $k\leq N$, it \emph{may happen} that the $k$:th labeled card is accepted. By this we mean that there is some sequence of cards for which with positive probability the selector's strategy accepts the $k$:th labeled card, but no earlier labeled card. 
If the selector's strategy didn't have this property, then the devil could switch to a strategy of always choosing exactly $k$ labeled cards, which would then win with probability 1.

Now suppose for a contradiction that at an equilibrium, the devil's strategy gives probability 0 to some number $k\leq N$. 
We want to show that the selector can change strategy and strictly increase his winning chances. There are two cases: If there is some number greater than $k$ which is chosen by the devil with positive probability, then let $l$ be the smallest such number. The change of strategy will consist in never accepting the $k$:th labeled card, and instead (whenever the original strategy dictated accepting the $k$:th labeled card) waiting for the $l$:th labeled card and accept that.

If on the other hand there is no number greater than $k$ which is chosen by the devil with positive probability, then let $m<k$ be the largest number which is. The change of strategy will now consist in always accepting the $m$:th labeled card. Since the original strategy sometimes waits for the $k$:th labeled card, this will give a strictly increased winning probability.

Contrary to assumption we are not at equilibrium, and this contradiction establishes the claim.
\end{proof}

\begin{Cor} \label{C:optimized}
At any equilibrium in the card game, the selector's strategy has the property that it wins with the same probability for all $n=1,\dots, N$.
\end{Cor}

We can assume without changing the players' winning chances that the selector, when deciding on the action on a labeled card, can take into account only the number of cards drawn earlier and the number of labeled cards among them, but not the actual positions (times) of the earlier labeled cards.
The selector's strategy can then be described by functions $f_1,\dots, f_N$, where $f_k:\{k,\dots, m\} \to [0,1]$, describes the probability $f_k(i)$ of accepting the $k$:th labeled card if it arrives as the $i$:th card turned up.

The next lemma shows that at equilibrium, each $f_k$ will have a sharp transition from 0 to 1 going via at most one value strictly between 0 and 1.  

\begin{Lemma} \label{L:sharp}
For $1\leq k \leq N$ and $k\leq i<m$, if $f_k(i) >0$, then $f_k(i+1) = 1$.
\end{Lemma}
\begin{proof}
If $k=N$ then $f_k$ must be identically 1 since by Corollary~\ref{C:optimized} the selector's strategy must optimize the winning chances against each possible value of $n$. 

Suppose therefore that $k<N$ and assume that $f_k(i)>0$. Then, having fixed the devil's strategy and optimizing for the selector, accepting the $k$:th labeled card at position $i$ is either better than rejecting, or indifferent. Since $k<N$, there is a positive probability that card $i+1$ is labeled. If it is, then the selector will have positive probability of winning if card $i$ is rejected, but obviously zero probability of winning if card $i$ is accepted. Conditioning on card $i+1$ being labeled, the selector is strictly better off rejecting card $i$ than accepting. Therefore conditioning on card $i+1$ \emph{not} being labeled, the selector is strictly better off accepting card $i$ than rejecting it. Now observe that the situation after drawing the $k$:th labeled card as card $i$, conditioning on card $i+1$ not being labeled, is the same as if drawing the $k$:th labeled card as card $i+1$. 

This shows that the selector must strictly prefer accepting the $k$:th card in position $i+1$ over rejecting it, which implies that $f_k(i+1)=1$. 
\end{proof}

\subsection{Reducing the original game to the card game}

We wish to draw conclusions about the original game, which in the following we call the \emph{last-arrival game}, from the analysis of the card game. 
Starting from the last-arrival game, the following is a sequence of games that become more and more advantageous to the selector:

\begin{enumerate}
\item[(1)] We choose a number $N$ and restrict the devil's choices to $1\leq n \leq N$.
\item [(2)] We choose a number $d$ and divide the unit interval into $d$ equal time-slots. We allow the selector to wait until the end of the current time-slot before deciding whether to accept a newly arrived item. 
\item[(3)] Moreover, we decide that the selector wins by default if any two items arrive in the same time-slot (even if the selector had previously accepted an item).
\item[(4)] Instead, we first generate $N$ independent arrival times $t_1,\dots,t_N$ and declare the selector winner if any two of \emph{those} fall in the same time-slot. If not, then after the devil chooses a number $n$, a random subset of $n$ of the times $t_1,\dots,t_N$ are chosen as arrival times of the $n$ items.
\item[(5)] Before the game, a biased coin is flipped, and with probability $\binom{N}2/d$, the selector is declared winner. If not, then the $(N, d)$-card game is played. 

\end{enumerate}

Notice that $\binom{N}2/d$ is an over-estimate of the probability that two of the numbers $t_1,\dots,t_N$ fall in the same time-slot. Moreover, conditioning on no such collision will reduce the game (4) to the card game.

\subsection{The restricted game}
The game in (1) above, which is continuous time but with the devil's choices restricted to $n\leq N$ will be called the \emph{$N$-restricted} game. By a \emph{threshold strategy} for the selector, we mean a strategy given by a sequence $a_1,\dots, a_N$, where the $k$:th item is accepted if and only if it arrives after time $a_k$.  

\begin{Prop}
For every $N$, there is a threshold strategy which is optimal for the selector in the $N$-restricted game.
\end{Prop}

\begin{proof}
We think of $N$ as fixed throughout the argument. Lemma~\ref{L:sharp} shows that for each $d$, the $(N,d)$-card game can be played optimally by what is essentially a threshold strategy in the discrete setting, the only small crux being that randomization may be needed if the $k$:th labeled card arrives exactly at the threshold position.

Let $a_k(d) = i/d$, where $i$ is the smallest value for which $f_k(i)>0$ in an optimal strategy for the $(N,d)$-card game as in Lemma~\ref{L:sharp}. 
By playing the restricted game using the threshold strategy given by $a_1(d),\dots,a_N(d)$ for large $d$, the selector will be able to achieve a winning probability in the restricted game approaching that of the game (5) in the list.

If we regard $a(d) = (a_1(d),\dots,a_N(d))$ as a point in the $N$-dimensional unit cube, then as $d\to\infty$, the points $a(d)$ must have an aggregation point, and by continuity of the selector's worst case expected payoff under a threshold strategy as a function of the thresholds, the aggregation point must correspond to an optimal strategy for the selector in the $N$-restricted game.
\end{proof}

Therefore, to obtain an upper bound on the selector's winning probability in the last-arrival game, it suffices to establish an upper bound for threshold strategies in the $N$-restricted game for some $N$.

\subsection{Canonical thresholds}
Suppose that for a fixed $N$ and some number $0<\theta<1$, we try to construct a threshold strategy that achieves a probability of success of at least $\theta$ for every $n\leq N$. Then there is a canonical way of recursively computing thresholds $a_1,\dots, a_N$ such that if any sequence of thresholds achieves winning probability at least $\theta$, then $a_1,\dots, a_N$ does. 

Suppose we have fixed $a_1,\dots, a_{k-1}$. Then if there is a value for $a_k$ that gives probability at leat $\theta$ of success for $n=k$, then there is a maximal such value of $a_k$, and for that value of $a_k$, the probability of success is exactly $\theta$.   

The choice of $a_k$ obviously does not affect the probability of success for $n<k$. For $n>k$, the probability of success will increase with larger $a_k$, since a larger $a_k$ decreases the risk of accepting an item too early. Therefore after fixing $a_1,\dots, a_{k-1}$, we might as well set $a_k$ to the unique value that gives success probability exactly $\theta$. If we do this consistently, starting with $a_1 = 1-\theta$, then we will either get stuck at some point because not even $a_k=0$ gives winning probability $\theta$ for $n=k$, or we find a sequence of thresholds that works.

For each $k$, the threshold $a_k$ (to the extent it can be defined) is a decreasing function of $\theta$, but does not depend on $N$. For each $N$ there is  
a maximal achievable success probability $\theta_N$, characterized by $a_N = 0$.
The sequence $\theta_N$ is decreasing with $N$, and therefore has a limit $\theta_{opt}$ as $N\to\infty$.

We can now establish the existence of an optimal strategy for the selector in the last-arrival game.

\begin{Thm}\label{T:selectorOpt}
There is a threshold strategy which is optimal for the selector in the last-arrival game.
\end{Thm}
\begin{proof}
For each $\theta>\theta_{opt}$ there is some $N$ for which $\theta_N < \theta$, which means that success probability $\theta$ is impossible even in some restricted games. On the other hand, for $\theta = \theta_{opt}$, the sequence $a_k$ of thresholds can be computed indefinitely without getting stuck, and therefore success probability $\theta_{opt}$ is achievable with a threshold strategy in the last-arrival (unrestricted) game.  
\end{proof}

\subsection{Optimal strategy for the devil}
The fact that the devil has an optimal strategy does not seem to follow from any ``soft'' argument. Conceivably the selector's task could become more difficult the larger the number of items. But we already have indications that such a situation would entail $\theta_{opt} = 1/e$. We will make this argument precise, and therefore the first step towards proving the existence of an optimal strategy for the devil is to show that contrary to the ``conclusion'' of the introduction, $\theta_{opt}$ is strictly smaller than $1/e$. An independent proof of this based on numerical calculation is given in Section~\ref{SS:smallerthan1/e}.

\begin{Prop} \label{P:thetaopt<1/e}
$$\theta_{opt} < 1/e.$$
\end{Prop}

\begin{proof} 
Clearly it suffices to consider the selector's optimal strategy. Let $a_k$ be the thresholds for this strategy. First observe that $a_k\to 1$ as $k\to\infty$. The reason is that if $a_k$ is too small, the selector will accept too early when $n=k+1$. It follows that there are infinitely many records in the sequence, by which we mean $a_k$ such that $a_k>a_i$ for every $i<k$. 

If $a_k$ is a record, then in order for the selector to succeed when $n=k$, there must be exactly one item arriving after time $a_k$. This shows that $\theta_{opt} \leq 1/e$, and if $\theta_{opt} = 1/e$, we must have $a_k = 1-1/k +o(1/k)$ whenever $a_k$ is a record. Now let $a_m$ and $a_n$ be records for large $m$, $n$ such that $n\geq2m$. Since $a_m = 1 - 1/m + o(1/m)$ and $a_n = 1-1/n+o(1/n)$, there must be a record $a_k$ for some $k$ in the interval $m<k\leq n$ such that $a_k-a_{k-1} \geq 1/k^2-o(1/k^2)$. 

Now consider the winning probability for $n=k$. The probability of exactly one item arriving in the interval $[a_k, 1]$ is maximized when $a_k = 1-1/k$ and therefore cannot be greater than $$\left(1-\frac1k\right)^{k-1} = \exp\left( (k-1)\log\left(1-\frac1k\right)\right) = e^{-1}\left(1+\frac{1}{2k}\right) + O\left(\frac{1}{k^2}\right),$$ as $k\to\infty$, by the Taylor expansion of $\exp((1/x-1)\log(1-x))$ where $x=1/k$.

Now condition on exactly one item arriving after time $a_k$. Then the strategy will still fail if there is an item arriving in the interval $[a_{k-1}, a_k]$. The probability of this is at least $1/k - o(1/k)$. Hence the selector's winning probability for $n=k$ is at most $$e^{-1}\left(1+\frac{1}{2k}-\frac{1}{k}\right) + o\left(\frac{1}{k}\right) < e^{-1}.$$ 
\end{proof}

\begin{Thm} \label{T:devilOpt}
The devil has an optimal strategy for the last-arrival game.
\end{Thm}

\begin{proof}
For every $N$, it follows by compactness of the set of strategies that there is an optimal strategy for the devil in the $N$-restricted game. The optimal strategy is simply a probability distribution on $\{1,\dots,N\}$. It is easy to verify that if a sequence of optimal strategies for the $N$-restricted games have an aggregation point in the total variation metric, then that aggregation point must represent a strategy for the last-arrival game which keeps the selector's success probability to $\theta_{opt}$, and which is therefore optimal for the devil.  

In order to prove the existence of such an aggregation point, it suffices, in view of the principle of dominated convergence, to prove that for every $\epsilon > 0$ there is a $C$ such that no optimal strategy of any $N$-restricted game assigns probability more than $\epsilon$ to the set of numbers $\{n:n>C\}$.

Obviously it suffices to find a $C'$ such that only finitely many $N$-restricted games have optimal strategies that give probability more than $\epsilon$ to $\{n:n>C'\}$, since we can then choose $C$ beyond the largest of those values of $N$.

Therefore assume for a contradiction that there is some $\epsilon>0$ such that for every $C$ there are arbitrarily large $N$ for which some optimal strategy in the $N$-restricted game gives probability $\epsilon$ or more to $\{n>C\}$. This means that, eliminating $N$ from the argument and considering the last-arrival game, for every $C$ and every $\theta>\theta_{opt}$, the devil can find a strategy that assigns probability at least $\epsilon$ to $\{n>C\}$ and keeps the selector's success probability to at most $\theta$.

We choose $$\theta = \theta_{opt} + \frac{\epsilon}2\cdot \left(\frac1e - \theta_{opt}\right).$$
By Proposition~\ref{P:thetaopt<1/e}, $\theta > \theta_{opt}$. Now for every $m_0$ and every $\delta>0$, we can find an $m\geq m_0$ and a strategy for the devil that keeps the selector's winning probability below $\theta$, and which assigns weight at most $\delta$ to $\{m+1,\dots, 2m\}$, and at least $\epsilon$ to $\{n: n>2m\}$. Indeed, if we choose $C\geq m_0\cdot 2^{1/\delta}$, then not all of the intervals $m_02^i<n\leq m_02^{i+1}$ for $i\leq 1/\delta$ can get weight more than $\delta$.

The idea is now to construct a strategy for the selector that ``sacrifices'' the values of $n$ in the interval $m<n\leq 2m$, and in return obtains success probability close to $1/e$ for all $n>2m$. This is done by first observing the number of items in the interval $[0, 1/2]$ (here the number $1/2$ is an arbitrary number smaller than every threshold of the selector's optimal strategy). Our strategy will never guess that $n$ is in the interval $m<n\leq 2m$, and the first step is to decide, based on the number of items arriving in $[0,1/2]$, between the two options $n\leq m$ and $n > 2m$.
Provided $\delta$ is small enough and $m_0$ (and thereby $m$) is large enough, this decision can be made with an error probability which is as small as we please. 

If we find that $n\leq m$, then we play normally, which means that conditioning on $n\leq m$ we achieve a winning probability as close as we please to $\theta_{opt}$. If on the other hand we find that $n>2m$, then we switch to the strategy used in the proof of Proposition~\ref{P:soft}, namely to accept the first item arriving after time $1-1/(2k)$, where $k$ is the number of items observed in $[0,1/2]$. This gives the selector a conditional probability of success close to $1/e > \theta$ in case $n$ actually is larger than $2m$. 

In all, this means that we can construct a strategy for the selector which achieves success probability as close as we please to $$(1-\epsilon) \cdot \theta_{opt} + \epsilon\cdot \frac1{e} = \theta_{opt} + \epsilon\cdot \left(\frac1e-\theta_{opt}\right) > \theta,$$
a contradiction.
\end{proof}

\section{Computing $\theta_{opt}$} \label{S:computing}
In this section we show how we have computed numerically the probability $\theta_{opt}$ of success for the selector with the optimal strategy, while providing the Maple code that was used. We obtain the bounds
$$0.3529170002071955 < \theta_{opt} < 0.3529170002071958.$$
\subsection{The $P$-polynomials}
We already described the idea of prescribing a value for $\theta$, which recursively defines thresholds that achieve success probability $\theta$ for the selector, if possible. Here we show how to actually compute these thresholds numerically. The calculations turn out to involve certain polynomials which we describe first. 

Let $[n] = \{1,\dots,n\}$. We define the polynomial $P_n$ in the variables $x_1,\dots,x_n$ by $$P_n = \sum_f x_{f(1)} x_{f(2)} \cdots x_{f(n)}.$$ where the sum is taken over all $f:[n]\to[n]$ such that for $1\leq i\leq n$, $$\left|f^{-1}\left([i]\right)\right| \geq i.$$ In other words, a monomial is present in the sum if and only if $x_1$ occurs at least once, $x_1$ and $x_2$ together occur at least twice etc. 

We have $P_0=1$, $P_1 = x_1$, $P_2 = x_1^2 + 2x_1x_2$, $P_3= x_1^3 + 3x_1^2x_2 + 3x_1^2x_3 + 3x_1x_2^2 + 6x_1x_2x_3$, etc. In general, the number of terms is equal to the $n$:th Catalan number, and the sum of the coefficients is equal to $(n+1)^{n-1}$. These polynomials are related to parking functions, trees etc, for which there is an extensive combinatorial theory.

The polynomials $P_n$ are calculated for $n\leq N$ (here we have taken $N=20$) using the Maple code:
\begin{verbatim} 
N:=20; 
P[0]:=1: 
for n to N do 
   P[n]:=
     add(binomial(n, n-i)*x[1]^(n-i)*eval(P[i], 
     {x[1]=add(x[k], k=2..n-i+1), seq(x[j]=x[j+n-i], j=2..i)}), 
     i=0..n-1): 
od:
\end{verbatim}

We must avoid using the command {\tt simplify} on $P_n$, since that will make each polynomial roughly four times the size (in terms of memory) of the previous one, instead of twice.

If $x_i$ are the lengths of disjoint subintervals $I_i$ of the unit interval, then $P_n$ is the probability that of $n$ independent uniformly chosen points, at least one lies in $I_1$, at least two in $I_1\cup I_2$, at least three in $I_1\cup I_2\cup I_3$ etc. 

We can see how this relates to the last-arrival problem if we assume that we have fixed the thresholds $a_1,\dots, a_n$ for a strategy, and we ask for the probability of not accepting any item if there are $n$ items. If we have $$a_1\leq a_2\leq\dots\leq a_n,$$ then not accepting any item is equivalent to having at least one item appearing before time $a_1$, at least two before $a_2$ etc, which has probability $$P_n(a_1,a_2-a_1,\dots, a_n-a_{n-1}).$$
To cover also the case that the $a_i$'s are not increasing, we let $$\alpha_{i,n} = \min_{i\leq j\leq n} a_j.$$ Then $\alpha_{1,n}\leq \alpha_{2,n}\leq \dots \leq \alpha_{n,n}$, and not accepting any of $n$ items is equivalent to having at least $i$ items appear before time $\alpha_{i, n}$ for every $i$, and therefore has probability $$P_n(\alpha_{1,n}, \alpha_{2,n}-\alpha_{1,n},\dots, \alpha_{n,n}-\alpha_{n-1,n}).$$

Suppose we fix a number $\theta$. Then we can recursively obtain upper bounds $b_i$ on the thresholds  $a_i$ for a strategy that achieves success probability at least $\theta$.
First we let $b_1 = 1-\theta$. This is easily seen to be an upper bound on $a_1$ by choosing $n=1$. Next we consider $n=2$ and derive an upper bound on $a_2$. We label a particular item and estimate the probability of success for the selector by $n$ times the probability that the labeled item is the last one to arrive and is accepted (and no earlier item is accepted).

For $n=2$ we obtain \begin{equation} \label{n2} P(\text{success}) \leq 2\cdot (1-a_2)\cdot a_1.\end{equation}
Success (with the labeled item) requires that (a) the labeled item arrives after time $a_2$ and (b) that the other item arrives before time $a_1$. Notice that we do not necessarily have equality: If $a_2<a_1$ then the labeled item need not be accepted even if (a) and (b) both hold. From \eqref{n2} it follows that $$a_2\leq 1-\frac{P(\text{success})}{2a_1}.$$ If now $P(\text{success})\geq \theta$ and $a_1\leq b_1$ it follows that $$a_2\leq 1-\frac{\theta}{2b_1},$$ so we can take $b_2= 1-\theta/(2b_1)$.
In general we have $$P(\text{success}) \leq n\cdot (1-a_n)\cdot P_{n-1}(\alpha_{1,n-1}, \alpha_{2,n-1}-\alpha_{1,n-1},\dots, \alpha_{n-1,n-1}-\alpha_{n-2,n-1}).$$

Suppose now that we have established upper bounds $b_i$ on $a_i$, and from them similarly define $$\beta_{i,n} = \min_{i\leq j\leq n} b_j.$$
Then \begin{multline} \theta \leq P(\text{success})\\ \leq n(1-a_n)\cdot P_{n-1}(\alpha_{1,n-1}, \alpha_{2,n-1}-\alpha_{1,n-1},\dots, \alpha_{n-1,n-1}-\alpha_{n-2,n-1}) \\ \leq n(1-a_n)\cdot P_{n-1}(\beta_{1,n-1}, \beta_{2,n-1}-\beta_{1,n-1},\dots, \beta_{n-1,n-1}-\beta_{n-2,n-1}),\end{multline} and therefore $$a_n\leq 1-\frac{\theta}{n\cdot P_{n-1}(\beta_{1,n-1}, \beta_{2,n-1}-\beta_{1,n-1},\dots, \beta_{n-1,n-1}-\beta_{n-2,n-1})}.$$
If we recursively define $$b_n = 1-\frac{\theta}{n\cdot P_{n-1}(\beta_{1,n-1}, \beta_{2,n-1}-\beta_{1,n-1},\dots, \beta_{n-1,n-1}-\beta_{n-2,n-1})},$$
then for every strategy that achieves success probability $\theta$ we must have $a_k \leq b_k$ for every $k$.

When we actually compute the bounds $b_k$, it is easiest to work with only one sequence of $\beta$'s that are continuously updated. Each round of computing a new $b_n$ is as follows:

$$\beta_n := 1-\frac{\theta}{n\cdot P_{n-1}(\beta_{1}, \beta_{2}-\beta_{1},\dots, \beta_{n-1}-\beta_{n-2})},$$

For $1\leq i\leq n-1$, $$\beta_i:=\min(\beta_i, \beta_n).$$

$$b_n:=\beta_n.$$

The Maple code is (here with the choice $\theta=1/e$):

\begin{verbatim}
theta:=evalf(exp(-1)); 
beta[0]:=0: 
for n to N+1 do 
  beta[n]:=1-theta/n/eval(P[n-1], 
        {seq(x[i]=beta[i]-beta[i-1], i=1..n-1)}): 
  b[n]:=beta[n]: 
  print(n, b[n]): 
  for i to n-1 do 
    if beta[i]>b[n] then beta[i]:=b[n]: fi: 
  od: 
od:
\end{verbatim}

\subsection{Upper bound smaller than $1/e$} \label{SS:smallerthan1/e}

If we start from the ansatz $\theta = 1/e$, we get the upper bounds 
\begin{eqnarray} \notag
                            b_1 = 0.632120558828558\\ \notag
                            b_2 = 0.709011646565337\\ \notag
                            b_3 = 0.753159994034917\\\notag
                            b_4 = 0.781938205897688\\\notag
                            b_5 = 0.801087317290966\\\notag
                            b_6 = 0.812758005512094\\\notag
                            b_7 = 0.816930118314804\\\notag
                            b_8 = 0.810454956408292\\\notag
                            b_9 = 0.780396988248096\\\notag
                            b_{10} = 0.657338420431837\\\notag
                            b_{11} = -1.22035626433979\notag
\end{eqnarray}
Here $b_{11}$ is negative, and since no strategy can have $a_{11}<0$, the conclusion is that no strategy for the selector can achieve winning probability $1/e$ for all $n$. In fact the conclusion is slightly stronger: Even in the restricted game for $N=11$, there is no way the selector can win with probability at least $1/e$. Therefore we have an alternative proof of Proposition~\ref{P:thetaopt<1/e}.

As long as the $b_i$'s are increasing, we actually obtain a strategy that wins with the desired probability, so we also see that if the devil is restricted to $1\leq n \leq 7$, the selector can in fact obtain winning probability $1/e$.

A little experimentation shows that for $\theta=0.35$, the selector doesn't seem to be in trouble, while the slightly larger $\theta=0.355$ leads to $b_{20}<0$. It also seems that as soon as the $b_i$ 's start decreasing, they enter a losing spiral that quickly leads to a negative value. 

\subsection{Refinement}
To obtain a more precise upper bound on $\theta_{opt}$, we would like to extend the computation of the $b_i$'s to $i>20$. With the method described so far, the problem is that the polynomials $P_n$ grow exponentially in terms of computer memory, preventing us to go much further than to $n=25$. 
Therefore we will use an estimate based on the polynomials $P_1,\dots, P_N$ for a fixed $N$, where in practice we have taken $N=24$.
We write $\alpha_i$ for $\alpha_{i,n-1}$ to make the equation readable:
\begin{multline} \label{ineq}
P(\text{success}) \leq n(1-a_n)\\ \cdot \sum_{i=0}^N \binom{n-1}{i} \alpha_{n-N-1}^{n-1-i}\cdot P_i(\alpha_{n-i}-\alpha_{n-N-1}, \alpha_{n-i+1}-\alpha_{n-i},\dots,\alpha_{n-1}-\alpha_{n-2}).
\end{multline}

The explanation is straightforward: Suppose that one of the $n$ items is labeled. The probability of success is $n$ times the probability that the labeled item is last and is accepted. Success with the labeled item requires that item to appear in the interval $[a_n, 1]$. Moreover, for some $i$ with $0\leq i\leq N$, we must have exactly $i$ items appearing after time $\alpha_{n-N-1, n-1}$. Of those, at least one must appear in $[\alpha_{n-N-1, n-1}, \alpha_{n-i, n-1}]$, at least two in $[\alpha_{n-N-1, n-1}, \alpha_{n-i+1, n-1}]$ etc.

We can use \eqref{ineq} to compute upper bounds $b_n$ in the same way as before. The Maple code is:

\begin{verbatim}
for n from N+2 to 1000 do 
  beta[n]:=1 - theta/n/add(binomial(n-1,i)*beta[n-N-1]^(n-1-i)
  *eval(P[i], {x[1] = beta[n-i] - beta[n-N-1], 
  seq(x[j] = beta[n-i-1+j] - beta[n-i-2+j], j=2..i)}), i=0..N):
  b[n]:=beta[n]: 
  print(n, b[n]): 
  for i to n-1 do 
    if beta[i]>b[n] then beta[i]:=b[n]: fi: 
  od: 
od: 
\end{verbatim}

With $N=24$ and $\theta=0.3529170002071958$ this gives $b_{147}<0$. Therefore \begin{equation} \label{upperBound} \theta_{opt}<0.3529170002071958.\end{equation}
%But $\theta=0.3529170002071957$ gives no contradiction.

\subsection{Lower bound}
So far we have not established any explicit lower bound on the selector's probability of success, although it seems that our upper bound is quite sharp. In this section we construct a threshold strategy for the selector that guarantees winning probability close to the upper bound \eqref{upperBound}. We restrict our attention to strategies for which the sequence of thresholds $a_k$ is nondecreasing.  Again we start from an ansatz for $\theta$. For $n\leq N+1$, we compute the sequence by
\begin{equation} \label{small} a_n = 1 - \frac{\theta}{nP_{n-1}(a_1, a_2-a_1,\dots,a_{n-1}-a_{n-2})}.\end{equation}

To extend the computations we now need a lower bound on the probability of success for the selector. Since we will require the sequence of $a_i$'s to be increasing, we do not have to bother with introducing the $\alpha$'s and $\beta$'s.

\begin{Prop}
\begin{multline} \label{lowerBound}
\frac{P({\rm success})}n \geq (1-a_n)\\ 
\cdot \sum_{i=0}^{N}\binom{n-1}ia_{n-N-1}^{n-1-i}\cdot P_i(a_{n-i}-a_{n-N-1},a_{n-i+1}-a_{n-i},\dots,a_{n-1}-a_{n-2})\\ 
\cdot \left(1-\sum_{k=N+2-i}^{n-1-i}\binom{n-1-i}k\left(1-\frac{a_{n-i-k}}{a_{n-N-1}}\right)^k\right). 
\end{multline}
\end{Prop}

\begin{proof} Success with a labeled item is equivalent to the following conditions being satisfied: \begin{enumerate} \item The labeled item appears in $[a_n, 1]$. \item For some $i$ with $0\leq i\leq N$, there appears exactly $i$ items in the interval $[a_{n-N-1},a_{n-1}]$, and the remaining $n-1-i$ appear before $a_{n-N-1}$.
\item Of the $i$ items in the interval $[a_{n-N-1}, a_{n-1}]$, at least one arrives in the interval $[a_{n-N-1}, a_{n-i}]$, at least two in the interval $[a_{n-N-1}, a_{n-i+1}]$ etc.
\item For the $n-1-i$ items that arrive before $a_{n-N-1}$, there is no subset of $k$ items for $N+2-i\leq k \leq n-1-i$ that arrive in the interval $[a_{n-i-k}, a_{n-N-1}]$. 
\end{enumerate}
Condition 1 corresponds to the factor $(1-a_n)$, condition 2 to the sum where $i$ goes from 0 to $N$, condition 3 to the occurrence of $P_i$, and condition 4 to the last factor in the summand. That last factor is where in general we do not have equality.
\end{proof}

For ``intermediate'' $n$, we use the recursion \begin{multline} \label{intermediate} a_n =\\ 1 - \frac{\theta}{n\sum_{i=0}^{N} \binom{n-1}{i}a_{n-N-1}^{n-1-i}\left(1-\sum_{k=N+2-i}^{n-1-i}\binom{n-1-i}{k}\left(1-\frac{a_{n-i-k}}{a_{n-N-1}}\right)^k\right)\cdot P_i},\end{multline} where $P_i$ is evaluated with the arguments as in \eqref{lowerBound}. The Maple code is 
\begin{verbatim}
for n from N+2 to 1000 do 
  a[n]:=1 - theta/n/add(binomial(n-1, i)*a[n-N-1]^(n-1-i)*
  (1 - add(binomial(n-1-i,k)*(1-a[n-i-k]/a[n-N-1])^k, 
  k=N+2-i..n-1-i))*
  eval(P[i], {x[1] = a[n-i] - a[n-N-1], 
  seq(x[j] = a[n-i-1+j] - a[n-i-2+j], j=2..i)}), i=0..N): 
  print(n, a[n]): 
  if a[n]<a[n-1] then print("Darn!"): fi:
od: 
\end{verbatim}

With $N=24$, this seems to work for $\theta=0.3529170002071955$.
%But not for $\theta=0.3529170002071956$.
Next we turn to proving rigorously that the strategy works for all $n$.

\subsection{Construction of a strategy}
We describe a scheme for proving that a certain sequence $a_k$ gives winning probability at least $\theta$. Our sequences will be of the following form: For some $m$, the numbers $a_1,\dots,a_m-1$ are given by a specific list, and for $k\geq m$, $$a_k=1-\frac1{k}.$$ 
We have taken $\theta= 0.3529170002071955$. If we compute $a_n$ according to the previous section, we notice an interesting phenomenon. The quantity $n(1-a_n)$ is at first increasing, and seems to stabilize at around $1.29$. Then for $n$ at around $130$ to $160$, it suddenly drops to around $0.74$ where it stabilizes.  

This phenomenon can be understood at least intuitively. For large $n$, the most important factor for success is that exactly one item must arrive after $a_n$. If the probability of success is exactly $\theta$, this confines $a_n$ to a short interval around one of the points $1-x_1/n$ and $1-x_2/n$, where $x_1$ and $x_2$ are the two solutions to $xe^{-x}=\theta$, unless some other reason for failure has unusually large probability. The second most important reason for failure is if an item arrives between $a_{n-1}$ and $a_n$. This means that the sequence $a_n$ can take large steps if it is on the hill between $1-x_1/n$ and $1-x_2/n$. The strategy is in danger of failure as long as $a_n\approx1-1.29/n$, but if it manages to climb the hill at $1-1/n$ it will reach security and stabilize at $a_n\approx1-0.74/n$. 

The first time $a_n>1-1/n$ is at $n=147$ and we therefore take $m=147$. In other words, we construct a sequence $a_n$ by \eqref{small} for $n\leq 25$, by \eqref{intermediate} for $26\leq n\leq 146$, and by $a_n = 1-1/n$ for $n\geq 147$. In computer language we overwrite the old values of $a_n$ for $n\geq 147$:

\begin{verbatim}
for n from 147 to 1000 do 
  a[n]:=1-1/n:
od:
\end{verbatim}

We have used \eqref{lowerBound} to verify that with this sequence, the probability of success is at least $\theta$ for $n\leq 750$.

\subsection{Proving a lower bound on $\theta_{opt}$}
Our next task is to prove that the constructed strategy wins with probability at least $\theta$ also for large $n$. Success requires that exactly one item arrives after time $a_n$. We start by estimating the probability of this. 
\begin{Lemma}
The probability of exactly one item appearing after time $a_n$ is at least $1/e$.
\end{Lemma}
\begin{proof}
We estimate it by $$n(1-a_n)a_n^{n-1} = \left(1-\frac{1}{n}\right)^{n-1} >  1/e.$$
\end{proof}

In the following we condition on exactly one item after time $a_n$. With a rough estimate, the (conditional) probability that there are at least $k$ (out of $n-1$) items arriving in the interval $[a_{n-k}, a_n]$ is at most $$\binom{n-1}{k}\left(1 - \frac{a_{n-k}}{a_n}\right)^k.$$
Conditioning on exactly one item arriving after time $a_n$, the failure probability is therefore bounded by \begin{equation} \label{conditionalBound} \sum_{k=1}^{n-1}\binom{n-1}{k}\left(1 - \frac{a_{n-k}}{a_n}\right)^k.\end{equation}

It remains to bound \eqref{conditionalBound}, which naturally we do in two steps. We begin with the terms for which $a_{n-k}$ is given by $1-1/(n-k)$. Recall that this is when $n-k\geq m$ and that in the specific case we have in mind, $m=147$.

\begin{Lemma}
For $n > m$, we have \begin{equation} \label{generalBound} \sum_{k=1}^{n-m}\binom{n-1}k\left(1-\frac{a_{n-k}}{a_n}\right)^k\leq \frac{e+1}{n-2}.\end{equation}
\end{Lemma}

\begin{proof}
For $n-k \geq m$ we have $$1-\frac{a_{n-k}}{a_n} = 1 - \frac{1-\frac{1}{n-k}}{1-\frac{1}{n}} = \frac{k}{(n-k)(n-1)}.$$
By an elementary integral estimate, $$\frac{k^k}{k!} \leq e^{k-1}$$ for every $k$. Therefore \begin{equation} \sum_{k=1}^{n-m}\binom{n-1}{k}\left(1 - \frac{a_{n-k}}{a_n}\right)^k \leq e^{-1}\sum_{k=1}^{n-m}\left(\frac{e}{n-k}\right)^k.\end{equation}
Differentiating the expression \begin{equation} \label{expr1} \left(\frac{e}{n-k}\right)^k\end{equation} twice with respect to $k$ gives \begin{equation}\left(\frac{e}{n-k}\right)^k\left(\log\left(\frac{e}{n-k}\right)+\frac{k}{n-k}\right)^2+\left(\frac{e}{n-k}\right)^k\left(\frac{2}{n-k}+\frac{k}{(n-k)^2}\right),\end{equation} which is obviously nonnegative. Hence \eqref{expr1} is convex regarded as a function of $k$. Therefore whenever \begin{equation} \label{conditionA}\frac{e^2}{(n-2)^2} \geq \left(\frac{e}{m}\right)^{n-m}, \end{equation} that is, when the first term is at least as large as the last term, we have $$e^{-1}\sum_{k=1}^{n-m}\left(\frac{e}{n-k}\right)^k \leq \frac{1}{n-1} + (n-m-1)\cdot\frac{e}{(n-2)^2}.$$ Since $$(n-2)^2\left(\frac{e}{m}\right)^n$$ decreases whenever $$n-2 \geq \frac{2}{-\log\left(\frac{e}{m}\right)},$$ it suffices to find a single value of $n$ for which \eqref{conditionA} is valid, and obviously equality holds when $n=m+2$.

Consequently $$\sum_{k=1}^{n-m}\binom{n-1}k\left(1-\frac{a_{n-k}}{a_n}\right)^k\leq \frac1{n-1}+\frac{e\cdot (n-m-1)}{(n-2)^2}$$ which in turn is $$\leq \frac{1}{n-2}+\frac{e\cdot(n-2)}{(n-2)^2} = \frac{e+1}{n-2}.$$ 
It is easy to check that this holds also when $n=m+1$. The bound \eqref{generalBound} now follows.
\end{proof}

The part of the sum in \eqref{conditionalBound} for which $k>n-m$ is estimated by (assuming that $n\geq 2m-3$) \begin{multline} \sum_{k=n-m+1}^{n-1} \binom{n-1}{k}\left(1-\frac{a_{n-k}}{a_n}\right)^k \leq \sum_{k=n-m+1}^{n-1} \binom{n-1}{m-2}(1-a_1)^k \\ \leq  (m-1)\cdot \frac{n^{m-2}}{(m-2)!}\cdot(1-a_1)^{n-m+1},\end{multline}
which is decreasing for $n\geq 2m$, since $$n\geq \frac{m-2}{-\log(1-a_1)}.$$ In our case it is smaller than $10^{-100}$.

Putting together these pieces, we obtain $$P(\text{success})\geq e^{-1} \cdot\left(1-\frac{e+1}{n-2} - 10^{-100}\right),$$
which clearly is larger than our choice of $\theta$ whenever $n>750$ (actually much earlier).
 We conclude that, with reservation for purely numerical errors, $$0.3529170002071955 < \theta_{opt} < 0.3529170002071958.$$

\section{The secretary problem for a partially ordered set} \label{S:poset}
The original motivation for considering the last-arrival problem was the secretary problem for a partially ordered set. The items (secretaries) are partially ordered and presented in random order to a selector, whose objective is to choose online a maximal element of the partial order. J.~Preater \cite{P99} showed that there is a strategy that achieves this with probability $1/8$, and this was improved to $1/4$ by N.~Georgiou, M.~Kuchta, M.~Morayne and J.~Niemiec \cite{G08} and further to $1/4+\epsilon$ by J.~Kozik \cite{K10}. Recently R.~Freij and the author \cite{FW10} found a strategy that achieves success probability $1/e$, thereby matching the upper bound from the classical secretary problem. 

Before finding the solution in \cite{FW10}, the analysis of several reasonable strategies showed that with decent probability they tend to accept the last maximal element to arrive. An idea was to consider a maximal element $x$ of the partial order, and to condition on all other maximal element arriving before $x$. As is shown below, this type of argument is strong enough to establish a nonzero probability of success in the partially ordered secretary problem.

Naturally this led to the conjecture that the correct generalization of the strategy for the classical secretary problem would turn out to accept the last maximal element of the partial order with probability at least $1/e$. 
For this conjecture it was natural to first consider a completely ``unordered'' set, and this is precisely the last-arrival problem. After some confusion, as described in the introduction, the conjecture turned out to be false, and therefore an analysis focusing only on the last arriving maximal element cannot establish winning probability $1/e$ for the partially ordered secretary problem. 

Nevertheless, it turns out that there is a strategy that works for an arbitrary partial order, and with probability uniformly bounded away from zero picks the last maximal element. The strategy builds on a strategy for the last-arrival problem, and uses the following result.

\begin{Lemma} There is a threshold strategy for the last-arrival problem which accepts the last item with probability at least $\theta_1>0$, and with probability at least $\theta_2>0$ does not accept any item. \end{Lemma} 

We have proved the first part of the statement in detail, and it is clear that the optimal strategy must also satisfy the second part (so does the strategy described in Proposition~\ref{P:soft}, but it is not a threshold strategy).

We now apply the following strategy for the partially ordered secretary problem: Accept an element if (1) it is maximal in the induced partial order of elements seen so far, and (2) arrives at a time later than $a_k$, where $k$ is the number of maximal elements in that induced partial order.

\begin{Prop} This strategy for the partially ordered secretary problem that with probability at least $\theta_1\theta_2>0$ picks the last maximal element. \end{Prop}

\begin{proof} We condition on the number $m$ of maximal elements in the partial order, and the number $q$ of elements that are dominated only by the last maximal element $x$ to arrive. We condition on $x$ arriving in the time interval $[a_m, 1]$. The probability for this is at least $\theta_1$.

With probability at least $\theta_2$, the $m+q-1$ maximal elements in the partial order $P - x$ will arrive early enough to guarantee that none of them, and thereby no element at all, is accepted before $x$ arrives (we have conditioned on the $m-1$ remaining maximal elements arriving before $x$ arrives, but this only increases the probability of success).
\end{proof}


\begin{thebibliography}{99}

\bibitem{B00} Thomas Bruss, \emph{Sum the Odds to One and Stop}, Annals of Probability {\bf 28} (2000), 1384--1391.

\bibitem{B03}  Thomas Bruss, \emph{A note on Bounds for the Odds-Theorem of Optimal Stopping}, Annals of Probability {\bf 31} (2003), 1859--1862.

\bibitem{FW10} Ragnar Freij and Johan W\"astlund, \emph{Partially ordered secretaries}, Electronic Communications in Probability {\bf 15} (2010), 504--507.%arXiv:1008.3310 (2010).

\bibitem{G08} Nicholas Georgiou, Malgorzata Kuchta, Michal Morayne and Jaroslaw Niemiec. \emph{On a universal best choice algorithm for partially ordered sets}, Random Structures and Algorithms, Volume 32 Issue 3, (May 2008), 263--273.
  
\bibitem{K10} Jakub Kozik, \emph{Dynamic threshold strategy for universal best choice problem}, DMTCS Proceedings, 21st International Meeting on Probabilistic, Combinatorial, and Asymptotic Methods in the Analysis of Algorithms (2010), 439--452. 
  
\bibitem{P99} John Preater, \emph{The best-choice problem for partially ordered objects}, Oper. Res. Lett. {\bf 25} (1999), 187--190.

\end{thebibliography}
\end{document}